\documentclass[submission,creativecommons]{eptcs}
\providecommand{\event}{CL\&C'18} 
\usepackage{underscore}           
\usepackage{amsmath}
\usepackage{amsthm}
\usepackage{amssymb}
\usepackage{stmaryrd}
\usepackage{microtype}
\usepackage{bussproofs}
\usepackage{xspace}
\usepackage[capitalise]{cleveref}
\usepackage{graphicx}
\usepackage[mathcal]{euscript}

\usepackage[draft]{fixme}
\fxsetup{inline,nomargin,marginclue,theme=color,index}

\title{Fast Cut-Elimination using Proof Terms:\\
An Empirical Study}
\author{Gabriel Ebner
\institute{TU Wien, Austria\thanks{Supported by the Vienna
Science and Technology Fund (WWTF) project \mbox{VRG12-004}.}}
\email{gebner@gebner.org}
}
\def\titlerunning{Fast Cut-Elimination Using Proof Terms}
\def\authorrunning{Gabriel Ebner}
\begin{document}
\maketitle

\newtheorem{theorem}{Theorem}
\newtheorem{lemma}[theorem]{Lemma}
\newtheorem{conjecture}[theorem]{Conjecture}

\newcommand\sbst{\backslash}
\newcommand\HasTy[1][]{\;::_{#1}\;}
\newcommand\bn{\hspace{-2.5pt}:}

\newcommand\newlktctr[1]{%
  \expandafter\newcommand\csname #1\endcsname{%
    \ensuremath{\mathrm{#1}}\xspace}}
\newlktctr{Cut}
\newlktctr{Ax}
\newlktctr{Rfl}
\newlktctr{TopR}
\newlktctr{NegL}
\newlktctr{NegR}
\newlktctr{AndL}
\newlktctr{AndR}
\newlktctr{AllL}
\newlktctr{AllR}
\newlktctr{Eql}
\newlktctr{Ind}

\newcommand\newnonterm[1]{%
  \expandafter\newcommand\csname #1\endcsname{%
    \ensuremath{\mathit{#1}}\xspace}}
\newnonterm{Hyp}
\newnonterm{Term}
\newnonterm{Expr}
\newnonterm{Var}
\newnonterm{Formula}
\newnonterm{Bool}

\newcommand\LKt{\ensuremath{\mathrm{LK}_\mathrm t}\xspace}

\newcommand\NN{\mathbb N}

\begin{abstract}
  Urban and Bierman introduced a calculus of proof terms for the sequent
  calculus LK with a strongly normalizing reduction relation
  in~\cite{Urban2001Strong}.  We extend this calculus to simply-typed
  higher-order logic with inferences for induction and equality, albeit
  without strong normalization.  We implement this calculus in GAPT, our
  library for proof transformations.  Evaluating the normalization on
  both artificial and real-world benchmarks, we show that this algorithm
  is typically several orders of magnitude faster than the existing
  Gentzen-like cut-reduction, and an order of magnitude faster than any
  other cut-elimination procedure implemented in GAPT.
\end{abstract}


\section{Introduction}

Cut-elimination is perhaps the most fundamental operation in proof
theory, first introduced by Gentzen in~\cite{Gentzen1935Untersuchungen}.
Its importance is underlined by a wide variety of its applications; one
application in particular motivates our interest in cut-elimination:
cut-free proofs directly contain Herbrand disjunctions.

Herbrand's theorem~\cite{Herbrand1930Recherches,Buss1995Herbrands}
captures the insight that the validity of a quantified formula is
characterized by the existence of a tautological finite set of
quantifier-free instances.  In its simplest case, the validity of a
purely existential formula $\exists x\: \varphi(x)$ is characterized by
the existence of a tautological disjunction of instances $\varphi(t_1)
\lor \dots \lor \varphi(t_n)$, a Herbrand disjunction.  Expansion proofs
generalize this result to higher-order logic in the form of elementary type
theory~\cite{Miller1987Compact}.

A computational implementation of Herbrand's theorem as provided by
cut-elimination lies at the foundation of many applications in
computational proof theory: if we can compress the Herbrand disjunction
extracted from a proof using a special kind of tree grammar, then we can
introduce a cut into the proof which reduces the number of quantifier
inferences---in practice this method finds interesting non-analytic
lemmas~\cite{Hetzl2012Towards, Hetzl2014Algorithmic,
Hetzl2014Introducing, Ebner2018generation}.  A similar approach can be
used for automated inductive theorem proving, where the tree grammar
generalizes a finite sequence of Herbrand
disjunctions~\cite{Eberhard2015Inductive}.  By comparing the Herbrand
disjunctions of proofs, we obtain a notion of proof equality that
identifies proofs which use the same quantifier
instances~\cite{Baaz2008CERES}.  Automated theorem provers typically use
Skolem functions; expansion proofs admit a particularly elegant
transformation that eliminates these Skolem functions and turns a proof
of a Skolemized formula into a proof of the original statement in linear
time~\cite{Baaz2012complexity}.  Herbrand disjunctions directly contain
witnesses for the existential quantifiers and hence capture a certain
computational interpretation of classical proofs.  Furthermore,
Luckhardt used Herbrand disjunctions to give a polynomial bound on the
number of solutions in Roth's theorem~\cite{Luckhardt1989Herbrand} in
the area of Diophantine approximation.

Our GAPT system for proof transformations contains implementations of
many of these Herbrand-based algorithms~\cite{Ebner2016System}, as well as
various proofs formalized in the sequent calculus LK and several
cut-elimination procedures.  However in practice we have proofs where
none of these procedures are successful, due to multiple reasons: the
performance may be insufficient, higher-order cuts cannot be treated,
induction cannot be unfolded, or special-purpose inferences such as
proof links are not supported.

The normalization procedure described in this paper has all of these features:
it is fast, supports higher-order cuts, can unfold induction inferences, and
does not fail in the presence of special-purpose inference rules.  This
procedure is based on a term calculus for LK described by Urban and
Bierman~\cite{Urban2001Strong}.  It is self-evident that proof normalization
can be implemented more efficiently using the Curry-Howard correspondence to
compute with proof terms instead of trees of sequents, as this significantly
reduces the bureaucracy required during reduction.  We also considered other
calculi such as the $\lambda\mu$-~\cite{Parigot1992Lambda} or the
$\lambda^{Sym}$-calculus~\cite{Barbanera1996Symmetric}.  In the end we decided
on the present calculus because of its close similarity to LK, as it
allows us to straightforwardly integrate special-purpose inferences.

In \cref{seccalculus} we present the syntax and typing rules for the
calculus as implemented in GAPT.  We then briefly describe the
implementation of the normalization procedure in \cref{secnorm}.  Its
performance is then empirically evaluated on both artificial and
real-world proofs in \cref{seceval}.  Finally, potential future work is
discussed in \cref{secfuture}.

One of the proofs on which we evaluate this normalization procedure
in \cref{seceval} is Furstenberg's famous proof of the infinitude of
primes~\cite{Furstenberg1955infinitude}.  Cut-elimination was also used
by Girard~\cite[annex 7.E]{Girard1987Proof} to analyze another proof of
Furstenberg that shows van der Waerden's theorem using ergodic
theory~\cite{Furstenberg1981Recurrence}.

\section{Calculus}\label{seccalculus}

The proof system is modeled closely after the calculus described in the
paper by Urban and Bierman~\cite{Urban2001Strong}.  Since the paper does
not give a name to the introduced calculus, we call our variant \LKt as
an abbreviation for ``LK with terms''.  Proofs in \LKt operate on
hypotheses (called names and co-names in~\cite{Urban2001Strong}), which
name formula occurrences in the current sequent.  We found it useful to
have a single type that combines both the names and co-names
of~\cite{Urban2001Strong} since it reduces code duplication.  Each
formula in a sequent is labelled by a hypothesis:
\vspace{-2ex}
\begin{prooftree}
  \AxiomC{$h_2\bn \varphi(t), h_1\bn \forall x\: \varphi(x), \Gamma
  \vdash \Delta$}
  \UnaryInfC{$h_1\bn \forall x\: \varphi(x), \Gamma \vdash \Delta$}
\end{prooftree}

Expressions in the object language are lambda expressions with simple
types: an expression is either a variable, a constant, a lambda
abstraction, or a function application.  Connectives and quantifiers
such as $\land^{o\to o\to o}$ and $\forall^{(\alpha\to o)\to o}_\alpha$
are represented as constants of the type indicated in the superscript.  Formulas are
expressions of type~$o$, which is the type of Booleans.  We identify
$\alpha\beta$-equal expressions.  A substitution $\sigma = [x_1 \sbst
s_1, \dots, x_n \sbst s_n]$ is a type-preserving map from variables to
expressions.  Given an expression $t$, we write $t\sigma$ for the
(capture-avoiding) application of the substitution $\sigma$ to $t$.

This language can express impredicative quantification over types of
arbitrary rank, such as predicates on predicates on functions: for
example $\forall f \forall C (C(Af) \to C(Bf)) \to \forall D\: (DA \to
DB)$ is syntactic sugar for $(\to) ((\forall) (\lambda f\: (\forall)
(\lambda C\: (\to)(C(Af))(C(Bf))))) ((\forall) (\lambda D\: (\to) (DA)
(DB)))$ where $A,B$ are predicates of type $(i\to i)\to o$ and the
quantifiers range over the variables $f^{i\to i}, C^{o \to o}$, and
$D^{((i\to i)\to o) \to o}$.  This formula expresses a form of
extensionality for such predicates, and is not provable in \LKt.

The proof terms are almost untyped: in contrast to
\cite{Urban2001Strong}, we include the cut formula in the proof term for
the cut inference to perform type-checking without higher-order unification.
A typing judgement then tells us what sequent a proof proves.
\Cref{lktsyntax} shows the syntax for the proof terms.  Hypothesis
arguments that are not bound are called main formulas: for example $h_1$
is a main formula of $\NegL(h_1, h_2\bn \pi)$.

\begin{figure}[t]
\begin{align*}
  \Hyp ::= {}& -\NN^+ \mid +\NN^+ \\
  \Term ::= {}& \Ax(\Hyp, \Hyp) \mid \TopR(\Hyp) \\
     \mid{}& \Cut(\Formula, \Hyp\bn \Term, \Hyp\bn \Term) \\
     \mid{}& \NegL(\Hyp, \Hyp\bn \Term)
     \mid \NegR(\Hyp, \Hyp\bn \Term) \\
     \mid{}& \AndL(\Hyp, \Hyp\bn \Hyp\bn \Term)
     \mid \AndR(\Hyp, \Hyp\bn \Term, \Hyp\bn \Term) \\
     \mid{}& \AllL(\Hyp, \Expr, \Hyp\bn \Term)
     \mid \AllR(\Hyp, \Var\bn \Hyp\bn \Term) \\
     \mid{}& \Rfl(\Hyp) \mid \Eql(\Hyp, \Hyp, \Bool, \Expr, \Hyp\bn \Term) \\
     \mid{}& \Ind(\Hyp, \Expr, \Expr,
       \Hyp\bn \Term,
       \Var\bn \Hyp\bn \Hyp\bn \Term)
\end{align*}\vspace{-5ex}
\caption{Syntax of \LKt}
\label{lktsyntax}
\end{figure}

We use named variables as a binding strategy for the hypotheses in
consistency with the implementation of the lambda expressions (as
opposed to de Bruijn indices or a locally nameless representation).
Hypotheses are stored as machine integers.  A negative hypothesis refers
to a formula in the antecedent, and a positive hypothesis refers to a
formula in the succedent of the sequent. The notation $\Hyp\bn \Term$
means that \Hyp is a bound variable in \Term, c.f.\ the notation of abstract
binding trees in~\cite{Harper2016Practical}. This encoding of LK is also
very similar to the encoding commonly used in logical frameworks (LF),
see~\cite{Pfenning1995Structural} for a description of such an approach.

Notably, there are no terms for weakening and contraction.  These are
implicit: we can use the same hypothesis zero or multiple times.  The
proof terms only contain new information that is not contained in the
end-sequent; only cut formulas, weak quantifier instance terms, and
eigenvariables are stored.  We do not repeat the formulas or atoms of
the end-sequent.

\begin{figure}[t]
\[
\AxiomC{}
\UnaryInfC{$\Ax(h_1, h_2)  \HasTy[\sigma]  h_1\bn \varphi, \Gamma \vdash \Delta, h_2\bn \varphi$}
\DisplayProof
\quad \quad \quad
\AxiomC{}
\UnaryInfC{$\TopR(h_1)  \HasTy[\sigma] \Gamma \vdash \Delta, h_1\bn \top$}
\DisplayProof
\]

\begin{prooftree}
\AxiomC{$\pi_1 \HasTy[\sigma] \Gamma \vdash \Delta, h_1\bn \varphi\sigma$}
\AxiomC{$\pi_2 \HasTy[\sigma] h_2\bn \varphi\sigma, \Gamma \vdash \Delta$}
\BinaryInfC{$\Cut(\varphi, h_1\bn \pi_1, h_2\bn \pi_2) \HasTy[\sigma]
  \Gamma \vdash \Delta$}
\end{prooftree}

\[
  \AxiomC{$\pi \HasTy[\sigma] h_1\bn\neg\varphi, \Gamma \vdash \Delta, h_2\bn \varphi$}
  \UnaryInfC{$\NegL(h_1, h_2\bn \pi) \HasTy[\sigma]
  h_1\bn\neg\varphi, \Gamma \vdash \Delta$}
  \DisplayProof
  \quad\quad\quad
  \AxiomC{$\pi \HasTy[\sigma] h_2\bn\varphi, \Gamma \vdash \Delta,
  h_1\bn\neg\varphi$}
  \UnaryInfC{$\NegR(h_1, h_2\bn \pi) \HasTy[\sigma]
  \Gamma \vdash \Delta, h_1\bn\neg\varphi$}
  \DisplayProof
\]

\begin{prooftree}
\AxiomC{$\pi  \HasTy[\sigma]  h_3\bn \psi, h_2\bn \varphi, h_1\bn \varphi \land \psi, \Gamma \vdash \Delta$}
\UnaryInfC{$\AndL(h_1, h_2\bn h_3\bn \pi)  \HasTy[\sigma]  h_1\bn \varphi \land \psi, \Gamma \vdash \Delta$}
\end{prooftree}

\begin{prooftree}
  \AxiomC{$\pi_1 \HasTy[\sigma] \Gamma\vdash\Delta, h_1\bn
  \varphi\land\psi, h_2\bn \varphi$}
  \AxiomC{$\pi_2 \HasTy[\sigma] \Gamma\vdash\Delta, h_1\bn
  \varphi\land\psi, h_3\bn \psi$}
  \BinaryInfC{$\AndR(h_1, h_2\bn\pi_1, h_3\bn\pi_2) \HasTy[\sigma]
  \Gamma \vdash \Delta, h_1\bn \varphi\land\psi$}
\end{prooftree}

\begin{prooftree}
  \AxiomC{$\pi \HasTy[\sigma] h_2\bn \varphi(t\sigma), h_1\bn \forall x\:
    \varphi(x), \Gamma \vdash \Delta$}
  \UnaryInfC{$\AllL(h_1, t, h_2\bn \pi) \HasTy[\sigma]
      h_1\bn \forall x\: \varphi(x), \Gamma \vdash \Delta$}
\end{prooftree}

\begin{prooftree}
\AxiomC{$\pi  \HasTy[{[x\sbst y]\sigma}]
  \Gamma\vdash\Delta,
  h_1\bn \forall x\: \varphi(x),
  h_2\bn \varphi(y)$}
\RightLabel{($y$ fresh)}
\UnaryInfC{$\AllR(h_1, x, h_2\bn \pi)  \HasTy[\sigma]
  \Gamma\vdash\Delta, h_1\bn \forall x\: \varphi(x)$}
\end{prooftree}

\[
  \AxiomC{}
  \UnaryInfC{$\Rfl(h) \HasTy[\sigma] \Gamma\vdash\Delta, h\bn t=t$}
  \DisplayProof
  \quad\quad\quad
  \AxiomC{$\pi  \HasTy[\sigma]
    h_1\bn t=s, \Gamma \vdash \Delta,
    h_2\bn \varphi\sigma(t), h_3\bn \varphi\sigma(s)$}
  \UnaryInfC{$\Eql(h_1, h_2, \mathtt{true}, \varphi, h_3\bn \pi) \HasTy[\sigma]
    h_1\bn t=s, \Gamma \vdash \Delta,
    h_2\bn \varphi\sigma(t)$}
  \DisplayProof
\]

\begin{prooftree}
  \AxiomC{$\pi_1 \HasTy[\sigma]
  \Gamma \vdash \Delta, h_1\bn \varphi(t)\sigma, h_2\bn \varphi\sigma(0)$}
  \AxiomC{$\pi_2 \HasTy[{[x\sbst y]}\sigma]
    \Gamma, h_3\bn \varphi\sigma(y) \vdash \Delta,
    h_1\bn \varphi(t)\sigma, h_4\bn \varphi\sigma(s(y))$}
  \RightLabel{($y$ fresh)}
  \BinaryInfC{$\Ind(h_1, \varphi, t, h_2\bn \pi_1, x\bn h_3\bn h_4\bn \pi_2)
    \HasTy[\sigma] \Gamma \vdash \Delta, h_1\bn \varphi(t)\sigma$}
\end{prooftree}

  \caption{Typing rules for \LKt}
  \label{lkttyping}
\end{figure}

Let us now define the typing judgment.  A local context is a finite map
from hypotheses to formulas.  We write $h_1\bn \varphi \vdash h_2\bn
\psi$ as a suggestive notation for the map $\{h_1 \mapsto \varphi, h_2
\mapsto \psi\}$ where $h_1$ is negative and $h_2$ is positive.  Outer
occurrences overwrite inner ones, that is $\vdash h\bn \varphi, h\bn
\psi$ means $\{h \mapsto \psi\}$.

Given an (expression) substitution $\sigma$, we can apply it to a proof
term $\pi$ in the natural way to obtain~$\pi\sigma$.  The judgment $\pi
\HasTy[\sigma] S$ means that $\pi\sigma$ is a valid proof in the local
context $S$, that is, $\pi\sigma$ proves that the sequent corresponding to $S$
is valid.  We may omit $\sigma$ if it is the identity substitution, in
this special case $\pi \HasTy \Gamma \vdash \Delta$ corresponds to the
notation $\Gamma \triangleright \pi \triangleright \Delta$ used
in~\cite{Urban2001Strong}.

The reason for parameterizing the typing judgement by a substitution is
twofold: due to our use of named variables, we may need to rename
bound eigenvariables (in \AllR and \Ind) when
traversing a term.  However, we do not want to apply a substitution to the
proof term to ensure that the eigenvariable is fresh.  This would be
both costly and also introduces an unnecessary dependency on the local
context in operations that would otherwise not require any typing
information.

The proof terms and corresponding typing rules are chosen in such a way that
they correspond as much as possible to the already implemented sequent calculus
LK, see~\cite[Appendix B.1]{GAPT2018User} for a detailed description of
that calculus.  The implementation also
contains further inferences for special applications, such as proof links for
schematic proofs~\cite{Cerna2018System}, definition rules~\cite{Baaz2006Proof},
and Skolem inferences to represent Skolemized
proofs in higher-order logic~\cite{Hetzl2011CERES}.  The implemented inference rule for
induction is also more general than the one shown here: it supports
structural induction over other types than natural numbers.

Equational reasoning is implemented using the \Rfl inference for
reflexivity, and an \Eql inference to rewrite in arbitrary contexts and
on both sides of the sequent.  The third argument indicates whether we
rewrite from left-to-right or right-to-left.  Syntactically, we support
equations between terms of arbitrary type, however cut-elimination can
fail with equations between functions or Booleans as quantified cuts can
remain.

In our version of higher-order logic, the connectives $\lor, \to, \bot$,
and $\exists$ are also primitive.  By a heavy abuse of notation, we
simply reuse the proof terms for $\land, \top,$ and $\forall$.  This
representation causes no confusion, since the intended connective is
always clear from the polarities of the hypotheses, and many operations
are defined identically for the different connectives.  The corresponding
typing rules are derived in the natural way, as an example we show the
case where \AndL is used to prove an implication on the right side:
\vspace{-2ex}
\begin{prooftree}
\AxiomC{$\pi  \HasTy[\sigma]   h_2\bn \varphi, \Gamma \vdash \Delta,
h_1\bn \varphi\to\psi, h_3\bn \psi$}
\UnaryInfC{$\AndL(h_1, h_2\bn h_3\bn \pi)  \HasTy[\sigma]  \Gamma \vdash
\Delta, h_1\bn \varphi\to\psi$}
\end{prooftree}

\section{Cut-normalization}\label{secnorm}

\newcommand\evalCut{\ensuremath{\mathcal E}\xspace}
\newcommand\normalize{\ensuremath{\mathcal N}\xspace}
\newcommand\subst{\ensuremath{\mathcal S}\xspace}

Normalization is performed in a big-step evaluation approach using 3
mutually recursive functions \normalize, \evalCut, and
\subst\footnote{In the implementation these are called
\texttt{normalize}, \texttt{evalCut}, and \texttt{ProofSubst}, resp.}.
All of these functions return fully normalized proof terms.
We do not create temporary \Cut terms, all produced \Cut terms are
irreducible (for example because they are ``stuck'' on \Eql or \Ind).
\Cref{figlktnorm} shows the definition of the functions
$\normalize,\evalCut$, and $\subst$.  Note that since contraction is
implicit, the cut rule behaves more like Gentzen's mix
rule~\cite{Gentzen1935Untersuchungen}.

\begin{itemize}

  \item The function \normalize takes a proof term $\pi$ as input and
    returns a normal form $\normalize(\pi)$.

  \item If $\pi_1$ and $\pi_2$ are already in normal form, then
    $\evalCut(\varphi, h_1\bn \pi_1, h_2\bn \pi_2)$ computes a normal
    form of $\Cut(\varphi, h_1\bn \pi_1, h_2\bn \pi_2)$.

  \item Let $\pi_1$ and $\pi_2$ be again in normal form, then
    $\subst(\pi_1, \varphi, h_1 := h_2\bn \pi_2)$ performs a proof
    substitution, which corresponds to the rank-reduction step of
    cut-elimination in LK.  The function $\subst$ takes one side of the
    cut and directly moves it to all inferences in the other side
    where the cut formula occurs as the main formula.  This operation is
    symmetric in the side of the cut, and only needs to be implemented
    once.

\end{itemize}

Given a term $\pi$, we write $\normalize(\pi) \downarrow$ if
$\normalize$ terminates on the input $\pi$; similarly for \evalCut
and \subst.

\begin{lemma}[Subject reduction]
  \[
    \def\defaultHypSeparation{\hspace{3pt}}
    \AxiomC{$\pi \HasTy[\sigma] \Gamma\vdash\Delta$}
    \AxiomC{$\normalize(\pi)\downarrow$}
    \BinaryInfC{$\normalize(\pi) \HasTy[\sigma] \Gamma\vdash\Delta$}
    \DisplayProof
    \hspace{1cm}
    \def\defaultHypSeparation{\hspace{3pt}}
    \AxiomC{$\pi_1 \HasTy[\sigma] \Gamma \vdash \Delta, h_1\bn \varphi\sigma$}
    \AxiomC{$\pi_2 \HasTy[\sigma] h_2\bn \varphi\sigma, \Gamma \vdash \Delta$}
    \AxiomC{$\evalCut(\varphi, h_1\bn \pi_1, h_2\bn \pi_2) \downarrow$}
    \TrinaryInfC{$\evalCut(\varphi, h_1\bn \pi_1, h_2\bn \pi_2) \HasTy[\sigma]
      \Gamma \vdash \Delta$}
    \DisplayProof
  \]
  \[
    \AxiomC{$\pi_1 \HasTy[\sigma] \Gamma\vdash\Delta, h_1\bn \varphi\sigma$}
    \AxiomC{$\pi_2 \HasTy[\sigma] h_2\bn \varphi\sigma, \Gamma\vdash\Delta$}
    \AxiomC{$\subst(\pi_1, \varphi, h_1 := h_2\bn \pi_2)\downarrow$}
    \TrinaryInfC{$\subst(\pi_1, \varphi, h_1 := h_2\bn \pi_2)
    \HasTy[\sigma] \Gamma\vdash\Delta$}
    \DisplayProof
  \]
  \[
    \AxiomC{$\pi_1 \HasTy[\sigma] h_1\bn \varphi\sigma, \Gamma\vdash\Delta$}
    \AxiomC{$\pi_2 \HasTy[\sigma] \Gamma\vdash\Delta, h_2\bn \varphi\sigma$}
    \AxiomC{$\subst(\pi_1, \varphi, h_1 := h_2\bn \pi_2)\downarrow$}
    \TrinaryInfC{$\subst(\pi_1, \varphi, h_1 := h_2\bn \pi_2)
    \HasTy[\sigma] \Gamma\vdash\Delta$}
    \DisplayProof
  \]
\end{lemma}
\begin{proof}
  Routine induction on the length of the computation of \normalize,
  \evalCut, \subst, resp.
\end{proof}

We expect that \normalize terminates on all well-typed \LKt proofs,
including higher-order quantifier inferences.  Urban and Bierman showed
strong normalization for their first-order calculus without
equality using reducibility methods~\cite{Urban2001Strong}.  \LKt is
more general as it is higher-order, and the first-order fragment is
slightly different due to the use of the skipping constructors
${\cdot}^?$, which skip unnecessary inferences.

\begin{conjecture}[Termination]
  Let $\pi \HasTy[\sigma] \Gamma\vdash\Delta$, then
  $\normalize(\pi)\downarrow$.
\end{conjecture}

Note that for our applications it is often not necessary to have
completely cut-free proofs.  Cuts on quantifier-free formulas are
for example unproblematic for the extraction of Herbrand disjunctions.

\begin{lemma}[Cut-elimination]
  Let $\pi \HasTy[\sigma] \Gamma\vdash\Delta$ such that
  $\normalize(\pi)\downarrow$.  If $\pi$ does not contain \Rfl, \Eql, or
  \Ind, then $\normalize(\pi)$ is cut-free.
\end{lemma}
\begin{proof}
  Cuts are only produced by \evalCut, and by case analysis
  this does not happen in this class.
\end{proof}

\begin{figure}
  \begin{align*}
    \normalize(\Cut(\varphi,h_1\bn\varphi, h_2\bn\psi)) &=
    \evalCut(\varphi, h_1\bn\normalize(\varphi), h_2\bn\normalize(\psi)) \\
    \normalize(\Ax(h_1,h_2)) &= \Ax(h_1,h_2) \\
    \normalize(\NegL(h_1,h_2\bn\pi)) &= \NegL^?(h_1,h_2\bn\normalize(\pi))
    \\\vdots & \quad\mbox{(other cases recurse analogously)}
  \end{align*}
  \begin{align*}
    \mbox{(if $h_1$ is not free in $\pi_1$:)}\quad
    \evalCut(\varphi, h_1\bn \pi_1, h_2\bn \pi_2) &= \pi_1 \\
    \mbox{(if $h_2$ is not free in $\pi_2$:)}\quad
    \evalCut(\varphi, h_1\bn \pi_1, h_2\bn \pi_2) &= \pi_2 \\
    \evalCut(\varphi, h_1\bn \Ax(h_2,h_1), h_3\bn \pi) &= \pi[h_3\sbst h_2] \\
    \evalCut(\varphi, h_1\bn \pi, h_2\bn \Ax(h_2,h_3)) &= \pi[h_1\sbst h_3] \\
    \evalCut(\neg\varphi, h_1\bn \NegR(h_1, h_2\bn \pi_1), h_3\bn \NegL(h_3, h_4\bn \pi_2)) &=
      \evalCut(\varphi, h_4\bn \pi_2', h_2\bn \pi_1') \\
    \evalCut(\varphi\land\psi,
    h_1\bn \AndR(h_1, h_2\bn \pi_1, h_3\bn \pi_2),
    h_4\bn \AndL(h_4, h_5\bn h_6\bn \pi_3)) &=
    \evalCut(\psi, h_3\bn \pi_2', h_6\bn \subst(\pi_3', \varphi, h_5 :=
    h_2\bn \pi_1')) \\
    \evalCut(\varphi \varoast \psi,
    h_1\bn \AndL(h_1, h_2\bn h_3\bn \pi_1),
    h_4\bn \AndR(h_4, h_5\bn \pi_2, h_6\bn \pi_3)) &=
    \evalCut(\psi,
    h_2\bn \subst(\pi_1', \varphi, h_3 := h_5\bn \pi_3'),
    h_5\bn \pi_2') \\
    \evalCut(\forall x\: \varphi(x),
    h_1\bn \AllR(h_1, y, h_2\bn \pi_1),
    h_3\bn \AllL(h_3, t, h_4\bn \pi_2)) &=
    \evalCut(\varphi(t),
    h_2\bn \pi_1[y\sbst t]',
    h_3\bn \pi_2') \\
    \evalCut(\exists x\: \varphi(x),
    h_1\bn \AllL(h_1, t, h_2\bn \pi_1),
    h_3\bn \AllR(h_3, y, h_4\bn \pi_2)) &=
    \evalCut(\varphi(t),
    h_2\bn \pi_1',
    h_4\bn \pi_2[y\sbst t]') \\
    \mbox{(if $h_1$ is not a main formula of $\pi_1$:)}\quad
    \evalCut(\varphi, h_1\bn \pi_1, h_2\bn \pi_2) &=
    \subst(\pi_1, \varphi, h_1 := h_2\bn \pi_2) \\
    \mbox{(if $h_2$ is not a main formula of $\pi_2$:)}\quad
    \evalCut(\varphi, h_1\bn \pi_1, h_2\bn \pi_2) &=
    \subst(\pi_2, \varphi, h_2 := h_1\bn \pi_1) \\
    \mbox{(otherwise:)}\quad
    \evalCut(\varphi, h_1\bn \pi_1, h_2\bn \pi_2) &=
    \Cut^?(\varphi, h_1\bn \pi_1, h_2\bn \pi_2)
  \end{align*}
  \begin{align*}
    \subst(\Ax(h_1,h_2), \varphi, h_1 := h_3\bn \pi) &=
    \pi[h_3 \sbst h_2] \\
    \subst(\Ax(h_1,h_2), \varphi, h_2 := h_3\bn \pi) &=
    \pi[h_3 \sbst h_1] \\
    \subst(\NegL(h_1, h_2\bn \pi_1), \varphi, h_3 := h_4\bn \pi_2) &=
    \NegL^?(h_1, h_2\bn \subst(\pi_1, \varphi, h_3 := h_4\bn \pi_2))
    \mbox{\quad (if $h_1 \neq h_3$)} \\
    \vdots & \quad\mbox{(other cases recurse analogously)}\\
    \subst(\pi_1, \varphi, h_1 := h_2\bn \pi_2) &=
    \evalCut(\varphi, h_1\bn \pi_1, h_2\bn \pi_2)
    \mbox{\quad (otherwise, if $h_1 < 0$)} \\
    \subst(\pi_1, \varphi, h_1 := h_2\bn \pi_2) &=
    \evalCut(\varphi, h_2\bn \pi_2, h_1\bn \pi_1)
    \mbox{\quad (otherwise, if $h_1 > 0$)}
  \end{align*}
  \begin{align*}
    \NegL^?(h_1, h_2\bn\pi) &= \pi \quad\mbox{(if $h_2$ not free in $\pi$)}
    \\
    \NegL^?(h_1, h_2\bn\pi) &= \NegL(h_1, h_2\bn \pi) \quad\mbox{(otherwise)}
    \\
    \vdots & \quad\mbox{(other cases analogously)}\\
  \end{align*}
  \caption{Evaluator for \LKt, where $\varoast \in \{{\lor},{\to}\}$.
    For reasons of space, we use the
  abbreviation $\pi_2'$ to abbreviate the proof substitution which
substitutes the opposite part of the cut for the main formula: for example, $\pi_2'$
abbreviates $\subst(\pi_2, \varphi, h_2 := h_1\bn \pi_1)$ in the case of
$\evalCut(\varphi, h_1\bn \pi_1, h_2\bn \NegL(h_2, h_3\bn \pi_2))$.}
\label{figlktnorm}
\end{figure}

We perform a few noteworthy optimizations:
\begin{itemize}

  \item Every term stores the set of its free hypotheses and free
    (expression) variables.  These are fields in the Scala classes
    implementing the proof terms.  We can hence efficiently (in
    logarithmic time) check whether a given hypothesis or variable is
    free in a proof term.

  \item Due to this extra data, we can effectively skip many calls of the
    normalization procedure.  We do not need to substitute or evaluate
    cuts if the hypothesis for the cut formula is not free in the
    subterm, in this case we can immediately return the subterm.

  \item When producing the resulting proof terms, we check whether we
    can skip any inferences.  For example, instead of $\NegL(h_1, h_2\bn
    \pi)$ we can directly return $\pi$ if $h_2$ is not free in $\pi$.
    In \cref{figlktnorm} we denote these ``skipping'' constructors with
    the ${\cdot}^?$ superscript.  This optimization is extremely important
    from a practical point of view, since it effectively prevents a
    common blow-up in proof size.

\end{itemize}

The cut-normalization in~\cite{Urban2001Strong} is presented as a
single-step reduction relation.  The strong normalization of that
relation depends on the fact that all cuts can be
eliminated in their calculus.  In \LKt however, cuts can be
irreducible---for example because they are stuck on an induction or on
\Eql.  This has the unfortunate consequence that the natural single-step
reduction relation for \LKt is not strongly normalizing.  Since multiple
cuts can be stuck on the same inference we have the traditional
counterexample of two commuting cuts, where for example
$\pi_3=\Eql(h_2,h_4,\mathtt{true}, \dots)$:
\[
  \Cut(\varphi, h_1\bn \pi_1, h_2\bn \Cut(\psi, h_3\bn \pi_2, h_4\bn \pi_3)) \:\mapsto\:
  \Cut(\psi, h_3\bn \pi_2, h_4\bn \Cut(\varphi, h_1\bn \pi_1, h_4\bn \pi_3)) \:\mapsto\: \dots
\]

\subsection{Induction unfolding}\label{secnormind}

We typically consider proofs with induction of sequents such as for
example $\forall x\: x+0=x, \forall x \forall y\: x+s(y)=s(x+y) \vdash
\varphi$ where $\varphi$ is quantifier-free (or maybe existentially
quantified), and the antecedent contains recursive definitions for all
contained function symbols such as (but not limited to) ${+}, {*}$, etc.
If $\varphi$ contains free variables or strong quantifiers, then we can
in general not eliminate all inductions---however the quantifier
instances of a normalized proof may still provide valuable insights.  In
particular we are interested in the quantifier instances of formulas in
the antecedent, as their structure plays an important role in our
approach to inductive theorem proving~\cite{Eberhard2015Inductive}.  The
language always contains the constructors $0$ and $s$.  Injectivity of
these constructors is included as an explicit formula in the antecedent
when necessary.  We consider arbitrary recursively defined functions,
also on other data types such as lists.

Elimination of induction inferences is handled in a similar way to
Gentzen's proof of the consistency of Peano Arithmetic~\cite{
Gentzen1936Widerspruchsfreiheit}.  Induction inferences whose terms are
constructor applications are unfolded:
\begin{align*}
  \Ind(h_1,\varphi,0,h_2\bn\pi_1,x\bn h_3\bn h_4\bn \pi_2)
  &\:\mapsto\:
  \pi_1[h_2\sbst h_1]
  \\
  \Ind(h_1,\varphi,s(t),h_2\bn\pi_1,x\bn h_3\bn h_4\bn \pi_2)
  &\:\mapsto\:
  \Cut(\varphi(t),
  h_1\bn \Ind(h_1,\varphi,t,h_2\bn\pi_1,x\bn h_3\bn h_4\bn \pi_2),
  h_3\bn \pi_2[x\sbst t][h_4\sbst h_1])
\end{align*}

The full induction-elimination procedure then alternates between
cut-normalization and full induction unfolding until we can no longer
unfold any induction inferences.  We also rewrite the
term $t$ in the induction inference using the universally quantified
equations representing the recursive definitions to bring the
term into constructor form.  The generated proof $\pi_s \HasTy \Gamma, h_5\bn
\varphi(t') \vdash h_4\bn \varphi(t)$ is then added via a cut, where
$t'$ is the simplified term which is now in constructor form:
\[
  \Ind(h_1,\varphi,t,\dots)
  \:\mapsto\:
  \Cut(\varphi(t'), h_4\bn \Ind(h_1,\varphi,t',\dots),
  h_5\bn \pi_s)
\]

We can perform this induction reduction even if the problem contains
function symbols that are not recursively defined.  In this case
inductions can remain in the output.  We conjecture that the full
induction-elimination procedure (alternating induction-unfolding and
cut-normalization) always terminates.

\subsection{Equational reduction}\label{secnormeql}

As noted in \cref{secnorm}, cuts on equational inferences are stuck.
Consider for example the following term, which cannot be reduced further:
\( \Cut(\forall x\: P(x,0), h_1\bn \Eql(h_1, h_2,
\mathtt{true}, \lambda y\: \forall x\: P(x,y), h_3\bn \pi_1), h_4\bn
\pi_2) \)

\newcommand\simEq{\ensuremath{\mathsf{sim}}\xspace}

This is clearly a problem since we cannot obtain Herbrand disjunctions
from proofs with such quantified cuts.  On the other hand, cuts on atoms
would pose no problem since we can still obtain Herbrand disjunctions by
examining the weak quantifier inferences.  We hence reduce quantified
equational inferences to
atomic equational inferences---then only atomic cuts can be stuck.

Concretely, we define a function $\simEq$ such that $\simEq(\varphi,
h_1, h_2, h_3) \HasTy[] h_1\bn l=r, h_2\bn \varphi(l) \vdash h_3\bn
\varphi(r)$ for any terms $l$ and $r$, where $\simEq$ only uses \Eql
inferences on atoms.  This function hence simulates \Eql inferences
using only \Eql inferences on atoms, and is straightforwardly defined by
recursion on $\varphi$.  We only show the case for conjunction as an
example:
\[ \simEq(\lambda x\: (\varphi(x) \land \psi(x)), h_1, h_2, h_3) =
  \AndL(h_2, h_4\bn h_5\bn \AndR(h_3, h_6\bn \simEq(\varphi, h_1, h_4,
h_6), h_7\bn \simEq(\psi, h_1, h_5, h_7))) \]

We then replace \Eql inferences on non-atoms using the translation $\Eql(h_1, h_2,
\mathtt{true}, \varphi, h_3\bn \pi) \mapsto \Cut(\varphi(r), h_4\bn
\simEq(\varphi, h_1, h_2, h_4), h_3\bn \pi)$.  Note that this
translation depends on the typing derivation (to obtain the term $r$ for
the cut formula) and can fail if we have equations between predicates.

\section{Empirical evaluation}\label{seceval}

\subsection{Artificial examples}\label{secevalartif}

\begin{figure}
  \includegraphics[bb=0 0 453.67775pt 590.33672pt]{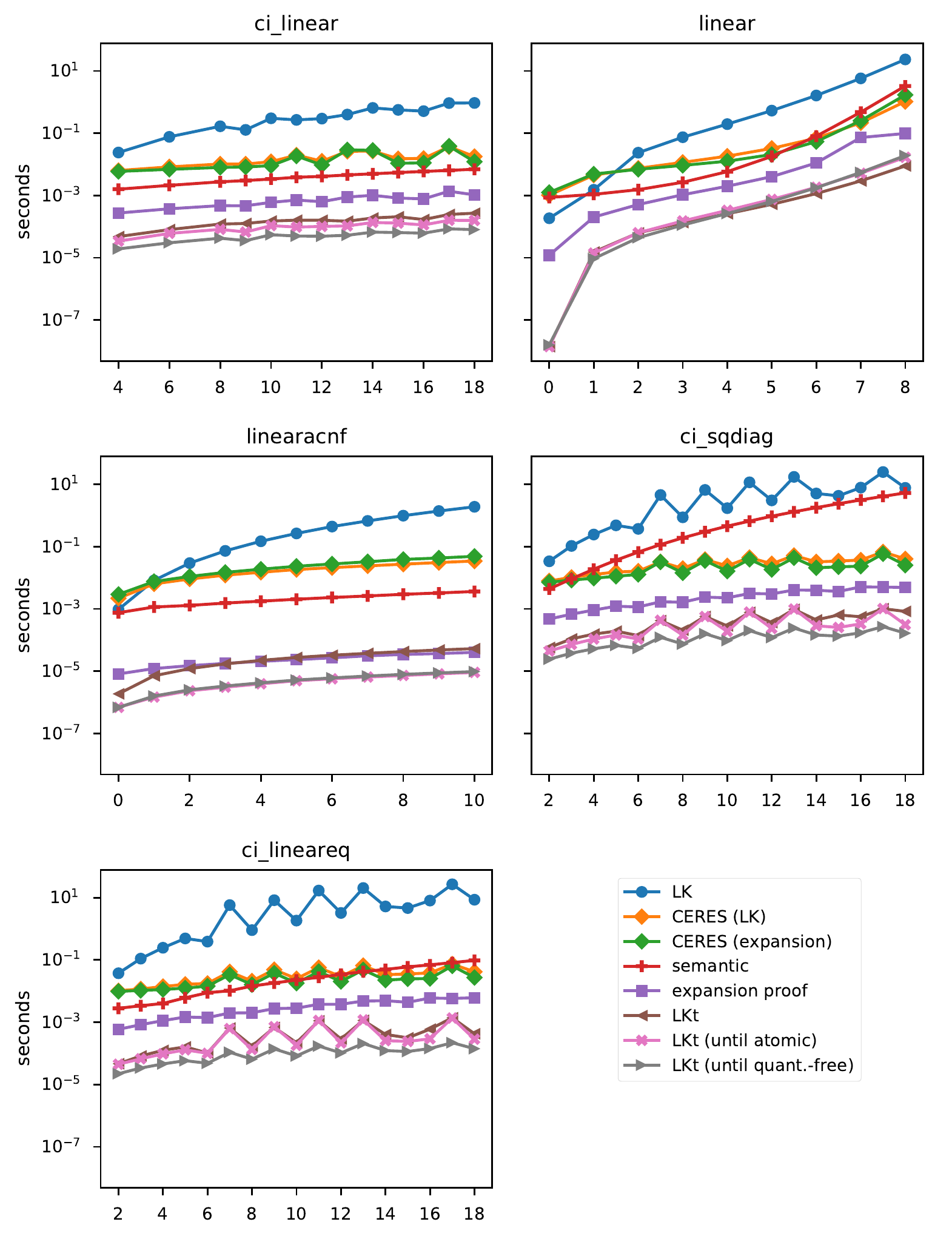}
  \caption{Runtime of cut-elimination procedures implemented in GAPT on
    several artificial examples.  (Some of the lines overlap: e.g. in
    \texttt{linearacnf}, the runtimes for \emph{LKt (until atomic)} and
    \emph{LKt (until quant.-free)} are almost identical.)}
  \label{figartif}
\end{figure}

The calculus and normalization procedure presented in this paper has
been implemented in the open source GAPT system\footnote{available at
\url{https://logic.at/gapt}} for proof
transformations~\cite{Ebner2016System}, version 2.10.  We now compare
the performance of several cut-normalization procedures implemented in
GAPT on benchmarks used in~\cite{Leitsch2018Extraction}.

\begin{itemize}

  \item \emph{LK:} Gentzen-style reductive cut-elimination in LK.  The
    proofs in LK are tree-like data structure where every node has a
    (formula) sequent.  The output is again a proof in LK, atomic cuts
    can appear directly below equational inferences.

  \item \emph{CERES (LK):} Cut-elimination by
    resolution~\cite{Baaz2000Cut} reduces the problem of cut-elimination
    in LK to finding a resolution refutation of a first-order clause
    set.  The output is a proof in LK with at most atomic cuts.

  \item \emph{CERES (expansion):} a variant of CERES that takes proofs
    with cuts in LK, and directly produces expansion
    proofs~\cite{Leitsch2018Extraction}.  This uses the same
    first-order clause sets as \emph{CERES (LK)}.

  \item \emph{semantic:} by ``semantic cut-elimination'', we refer to
    the procedure that throws away the input proofs, and generates a
    cut-free proof from scratch.  GAPT contains interfaces to several resolution
    provers, including the built-in Escargot prover.  Here we used
    Escargot to obtain a cut-free expansion proof of the end-sequent of
    the input proof.

  \item \emph{expansion proof:} the expansion proofs implemented in GAPT
    support cuts---such cuts corresponds to cuts in LK and are simply
    expansions of the formula
    $\forall X\: (X \to X)$.  First-order cuts in expansion proofs can
    eliminated using a procedure described in~\cite{Hetzl2013Expansion},
    which operates just on the quantifier instances of the proof,
    and is similar to the proofs of the epsilon
    theorems~\cite{Hilbert1939Grundlagen}.  Both the input and output
    formats are expansion proofs, the resulting expansion proof is
    cut-free.

  \item \emph{LKt:} the normalization procedure shown in \cref{secnorm}.

  \item \emph{LKt (until atomic):} same as \emph{LKt}, but we do not
    reduce atomic cuts.  The resulting proof may still contain cuts on
    atoms, but this is sufficient for the extraction of Herbrand
    disjunctions.  We can directly extract Herbrand disjunctions from
    proofs as long as all cut formulas are propositional.

  \item \emph{LKt (until quant.-free):} same as \emph{LKt}, but we do
    not reduce quantifier-free cuts.

\end{itemize}

The graphs in \cref{figartif} show the runtime for each of these
procedures on several artificial example proofs.  The runtime is measured in
seconds of wall clock time; we used a logarithmic scale for the time
since the performance of the procedures differs by several orders of
magnitude.  In one case, \emph{LKt (until quant.-free)} is 1000000 times
faster than \emph{LK}.  All of the example proofs are parameterized by a
natural number $n \geq 0$ (the x-axis of the plot), the size of the
input proofs is polynomially bounded in $n$.

\paragraph{Linear example after cut-introduction (\texttt{ci\_linear})}

The name ``linear example'' refers to the sequence of (proofs of) the
sequent $P(0), \forall x\: (P(x) \to P(s(x))) \vdash P(s^n(0))$.  We
take natural cut-free proofs of this sequent and then use an automated
method that introduces universally quantified
cuts~\cite{Ebner2018generation} to obtain a proof with cut.  In GAPT, these
proofs with universally quantified cuts are produced with
\verb|CutIntroduction(LinearExampleProof(n)).get|.

In this example, all of the \LKt normalization procedures are faster
than the CERES variants by a factor of about 100x.  Even semantic
cut-elimination is faster.  \LKt normalization is also faster than
expansion proof cut-elimination by a factor of about 10x.  We also see
that not eliminating atomic cuts is a bit faster than full
cut-elimination, and not eliminating quantifier-free cuts is even
faster.

\paragraph{Linear example proof with manual cuts (\texttt{linear})}

Cut-introduction often produces unnecessarily complicated lemmas,
resulting in irregularity when used in proof sequences.  It is also
limited to small proofs. To produce a more regular sequence and obtain
larger proofs, we manually formalized natural proofs of the linear
example for $2^n$ using $n-1$ cuts with the cut formulas $\forall x\:
(P(x) \to P(s^{2^k}(x)))$ for $1 \leq k < n$.  These proofs can be
obtained with \verb|LinearCutExampleProof(n)|.  (Note that this sequence
of proofs produces exponentially larger cut-free proofs than the other
sequences.)

The results are similar to the proofs obtained with cut-introduction,
although we observe new phenomena at both ends of the sequence: for
$n=0$, the proofs consist of a single axiom.  Here, the \LKt-based
procedures produce a cut-free proof in about 15 nanoseconds.  On the
other end, at $n \geq 6$, we finally see CERES becoming slightly faster
than semantic cut-elimination.

\paragraph{Linear example proof with atomic cuts (\texttt{linearacnf})}

To complete the discussion of the linear example, we also consider
a proof sequence in atomic cut-normal form (ACNF).  In these proofs, the
quantifier and propositional inferences are on the top of the proof, and
the bottom part consists only of atomic cuts---very much like a ground
resolution refutation.  Interestingly, atomic cut-elimination is
surprisingly cheap in this example: the \LKt-based normalization only
takes 10 microseconds.  On the other hand, the CERES-based methods
require as much time as they do for the proofs with universally
quantified cuts: they refute a clause set whose size is linear in $n$.

\paragraph{Square diagonal proof after cut-introduction
(\texttt{ci\_sqdiag})}

Just as in the linear example, we take cut-free proofs of
$P(0,0), \forall x \forall y\: (P(x,y) \to P(s(x), y)), \forall x
\forall y\: (P(x,y) \to P(x,s(y))) \vdash P(s^n(0), s^n(0))$ and then
automatically introduce universally quantified cuts.
\LKt
normalization until quantifier-free cuts is an order of magnitude faster
than expansion proof cut-elimination, and two orders of magnitude faster
than CERES.

\paragraph{Linear equality example proof after cut-introduction
(\texttt{ci\_lineareq})}

These proofs are generated using
\texttt{CutIntroduction(LinearEqExampleProof(n))}. Note that we replaced
the equality predicate by a binary E relation to prevent accidental
introduction of equational inferences. Again, \LKt normalization until
propositional cuts is 10x faster than expansion proof cut elimination,
which is 10x faster than CERES.

The astute reader will have noticed the spikes in the runtime of the
reductive cut-elimination procedures at $n \in \{7,9,11,13,17\}$. These
spikes are due to convoluted cut formulas produced by cut-introduction.
For example at $n=17$, the cut formula is $\forall x\: ((E(f(x), a) \to
E(f^3(x), a)) \land (E(x, a) \to E(f(x), a)))$ and we use it to prove
$E(a,a) \to E(f^{17}(a), a)$---this proof is almost as complicated on
the propositional level as the cut-free proof, even though it has a
lower quantifier complexity.

\subsection{Mathematical proofs}

\begin{figure}[t]
  \includegraphics[scale=0.97, bb=0 0 453.67775pt 372.08134pt]{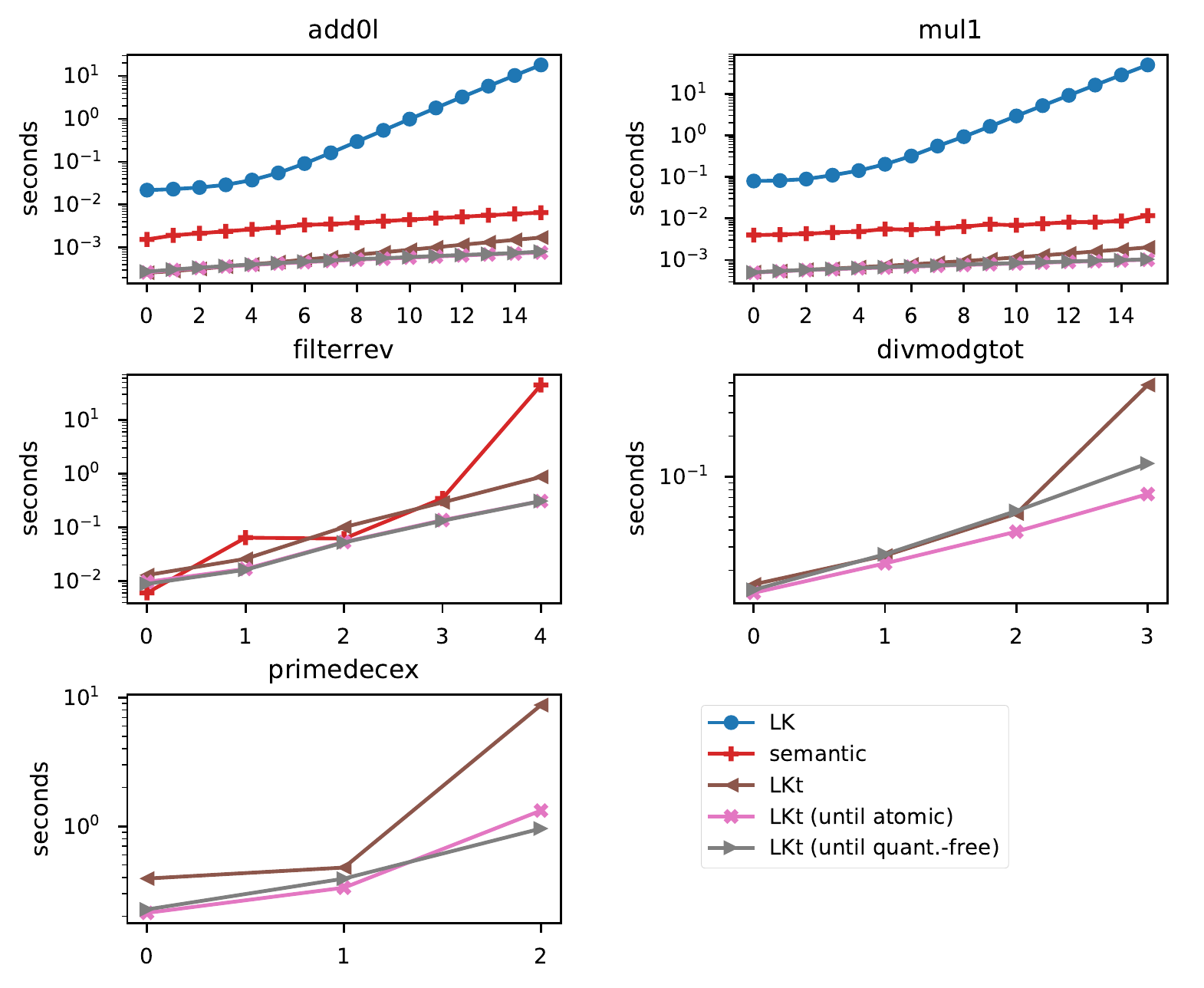}
  \caption{Runtime of induction-elimination procedures on
  lemmas formalized in GAPT's library.  (Results are omitted for some
  methods due to excessive runtime.)}
  \label{figmath}
\end{figure}

GAPT contains a small library of formalized proofs for testing.  These
are mostly basic properties of natural numbers and lists.  The biggest
formalized result is the fundamental theorem of arithmetic, showing the
existence and uniqueness of prime decomposition for natural numbers.
We evaluated the performance of \LKt as well as other procedures (see
\cref{secevalartif}) on several of these proofs.
\Cref{figmath} shows the runtime of the induction- and cut-elimination.
We tested instances of proofs of the following statements:

\vspace{1em}
\begin{tabular}{rl}
  \texttt{add0l} & $\forall x\: (0+x = x)$ \\
  \texttt{mul1} & $\forall x\: (x*1 = x)$ \\
  \texttt{filterrev} & $\forall p\: \forall l\:
  (\mathrm{filter}(p, \mathrm{rev}(l)) =
  \mathrm{rev}(\mathrm{filter}(p,l)))$ \\
  \texttt{divmodgtot} & $\forall a\: \forall b\:
  (b \neq 0 \to \exists d\: \exists r\: (r<b \land d*b+r=a))$ \\
  \texttt{primedecex} & $\forall n\:
  (n \neq 0 \to \exists d\: \mathrm{primedec}(d,n))$
\end{tabular}
\vspace{1em}

The proofs contain the primitive recursive definitions in the
antecedent.  For example, \texttt{add0l} is a proof with induction of
the sequent
\(
  \forall x\: (x+0=x),\: \forall x\: \forall y\: (x+s(y)=s(x+y))
  \:\vdash\:
  \forall x\: (0+x = x)
  \).
As before, induction-elimination in LK is several orders of magnitude
slower than in \LKt.  Semantic cut-elimination is surprisingly fast, it
is as fast as \LKt for small instances of \texttt{filterrev}.

\subsection{Furstenberg proof}\label{secfurstenberg}

\begin{figure}[b]
  \includegraphics[bb=0 0 394.79054pt 156.57057pt]{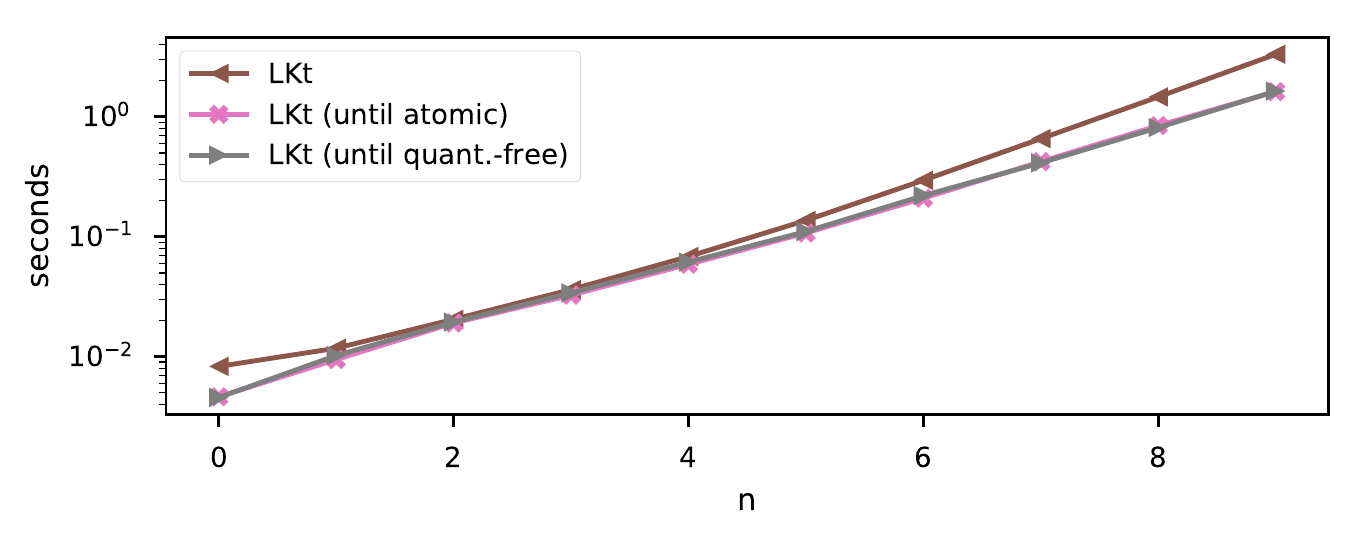}
  \vspace{-1em}
  \caption{Normalization runtime on a formalization of
  Furstenberg's proof of the infinitude of primes.}
  \label{figfurstenberg}
\end{figure}

Furstenberg's well-known proof of the infinitude of
primes~\cite{Furstenberg1955infinitude} equips the integers
with a topology generated by arithmetic progressions, and uses this
machinery to show that there are infinitely many primes.  For every
natural number $n \in \mathbb N$ we hence get a second-order proof
$\pi_n$ showing that there are more than $n$ prime numbers.  Cut-elimination of
$\pi_n$ then extracts the computational content of Furstenberg's
argument: we get a new prime number as a witness.

CERES was used to perform this extraction manually~\cite{Baaz2008CERES}.
The key step in cut-elimination using CERES consists of the refutation
of a so-called characteristic clause set.  In the case of Furstenberg's
proof automated theorem provers could only refute this first-order
clause set for $n=0$, that is, to show that there is more than one prime
number.  The authors hence manually constructed a sequence of
refutations, taking Euclid's proof of the infinitude of primes as a
guideline to obtain a prime divisor of $1 + p_0 \dots p_n$ as a witness.
The authors also present another refutation for $n=2$, which yields more
than one witness: one of the numbers $p_0+1$, $p_1+1$, or $5$ contains
the third prime as a divisor.

Using \LKt, GAPT can now perform the cut-elimination and extract the
witness term automatically.  \Cref{figfurstenberg} shows the performance
of the \LKt~normalization on instances of Furstenberg's proof.  The
concrete formalization closely resembles the one described
in~\cite{Baaz2008CERES}, however there have been minor changes to
account for subtle differences in the LK calculus currently implemented
in GAPT.  Now that we could
cut-eliminate this particular formalization for the first time, we were
excited to find an interesting feature.  We expected that the
cut-elimination of Furstenberg's proof would compute the same witness as
Euclid's proof: a prime divisor of $p_0 \cdots p_n + 1$.  However we got
the following witness instead, which contains $2$ as an additional
factor:
\( \mathtt{primediv\_of}(1 + 2 * p(0) * \cdots * p(n)) \)

This constant factor seems to depend on the concrete way in which we
formalize the lemma that nonempty open sets are infinite (this lemma is
called $\varphi_2$ in~\cite{Baaz2008CERES}).  With a slightly different
quantifier instance there, we can also get $3$ instead of $2$.

\section{Future work}\label{secfuture}

As our focus here lies in the practical applications of cut-elimination,
termination of the \LKt normalization procedure is only of secondary
concern.  For an actual implementation, there is little difference
between an algorithm that does not terminate or one that terminates
after a thousand years---as long as it quickly terminates on the
instances we apply it to.  For some classes of proofs, it is
straightforward to see that normalization indeed always terminates.  Due
to the direct correspondence with the traditional presentation of LK, we
can reuse termination arguments.  Whenever we observe non-termination in
the \LKt normalization, we get a corresponding non-terminating reduction
sequence in LK with an uppermost-first strategy.  We believe that
induction unfolding can be shown to terminate via a similar
argument as used in~\cite{Takeuti1987Proof} for the proof of the
consistency of Peano Arithmetic.  It remains open whether normalization
terminates for proofs with higher-order (or even just second-order)
quantifier inferences.

The current handling of equational and induction inferences as described
in \cref{secnormeql,secnormind} is unsatisfactory as they are not
integrated in the main normalization function but require a separate
pass over the proof.  Furthermore, the normalized proof may contain \Eql
inferences on atoms.  We are not aware of any terminating procedure
using local rewrite rules that eliminates unary equational inferences
such as the ones used in \LKt.

Renaming hypotheses and applying expression substitutions incurs a
significant cost in the benchmarks.  An obvious solution is to introduce
an explicit substitution inference to implement these operations without
the need to traverse the proof term.  In fact, one of the motivations
behind the substitution parameter $\sigma$ in the typing judgment
$\HasTy[\sigma]$ was the support for explicit substitution inferences.

As a cheap optimization, we could grade-reduce blocks of quantifier
inferences in a single substitution.  This should speed up the common
case of eliminating lemmas with many universal quantifiers.  Another
possible optimization is the use of caching: all of the functions
$\normalize, \evalCut,$ and $\subst$ are pure, making it easy to cache
their results.  However in practice caching seems to degrade
performance: simply caching the result of~\evalCut causes a 10-20\%
increase in runtime on the benchmarks of \cref{secevalartif}.
Normalization problems in \LKt do not seem to repeat often enough to
warrant a cache.

We used named variables as a binding strategy since this is
traditionally used in GAPT.  As expected, this choice has resulted in
a number of overbinding-related bugs, which were difficult to debug.
However, with named variables we can often avoid renaming
when traversing and substituting proofs---where other approaches such as
de Bruijn indices or locally nameless would always require
renaming, or instantiation and abstraction, resp.  Since every term in \LKt contains
(multiple) binders, it seems prudent to avoid renaming in the common
case.  It may be possible to implement an efficient binding strategy
using de~Bruijn indices or a locally nameless representation by adding
explicit renaming inferences.

Proof assistants such as Lean, Coq, or Minlog also provide functions to
normalize proofs.  It would be interesting to compare their performance
to the approaches implemented in GAPT.

\section{Conclusion}

Term assignments to proofs provide an elegant
implementation technique for the efficient computation and
transformation of proofs.  We have
obtained a speed-up of several orders of magnitude just by switching the
representation from trees of sequents to untyped proof terms.
The normalization procedure implemented in this paradigm and described
in this paper is fast, supports higher-order cuts, can unfold induction
inferences, and can normalize cuts in the presence of all inference
rules supported by GAPT.  As shown in \cref{secfurstenberg}, we can now
practically cut-eliminate proofs which were out of reach before.

However our ultimate interest lies in the quantifier structure of (cut-free)
proofs as captured by Herbrand disjunctions or (in general) expansion
proofs.  From this point of view, we are not restricted to
cut-elimination in LK or inessential variations like \LKt.
Another option that is radically different from what we have considered
so far is to use functional interpretation to compute expansion proofs
as described by Gerhardy and Kohlenbach in~\cite{Gerhardy2005Extracting}.

For proofs with only universally quantified first-order cuts, a certain
type of tree grammar describes the quantifier
inferences~\cite{Hetzl2012Applying}, and the language generated by such
a grammar then directly corresponds to a Herbrand sequent.  We plan to
develop and implement extensions of this grammar-based approach to
general first-order cuts for an efficient extraction of Herbrand
disjunctions, see~\cite{Afshari2018Herbrand} for grammars describing
general prenex cuts.

\bibliographystyle{eptcs}
\bibliography{fastlkt}

\begin{thebibliography}{10}
\providecommand{\bibitemdeclare}[2]{}
\providecommand{\surnamestart}{}
\providecommand{\surnameend}{}
\providecommand{\urlprefix}{Available at }
\providecommand{\url}[1]{\texttt{#1}}
\providecommand{\href}[2]{\texttt{#2}}
\providecommand{\urlalt}[2]{\href{#1}{#2}}
\providecommand{\doi}[1]{doi:\urlalt{http://dx.doi.org/#1}{#1}}
\providecommand{\bibinfo}[2]{#2}

\bibitemdeclare{article}{Afshari2018Herbrand}
\bibitem{Afshari2018Herbrand}
\bibinfo{author}{Bahareh \surnamestart Afshari\surnameend},
  \bibinfo{author}{Stefan \surnamestart Hetzl\surnameend} \&
  \bibinfo{author}{Graham \surnamestart Leigh\surnameend}
  (\bibinfo{year}{2018}): \emph{\bibinfo{title}{Herbrand's Theorem as Higher
  Order Recursion}}.
\newblock \doi{10.14760/OWP-2018-01}.

\bibitemdeclare{inproceedings}{Baaz2006Proof}
\bibitem{Baaz2006Proof}
\bibinfo{author}{Matthias \surnamestart Baaz\surnameend},
  \bibinfo{author}{Stefan \surnamestart Hetzl\surnameend},
  \bibinfo{author}{Alexander \surnamestart Leitsch\surnameend},
  \bibinfo{author}{Clemens \surnamestart Richter\surnameend} \&
  \bibinfo{author}{Hendrik \surnamestart Spohr\surnameend}
  (\bibinfo{year}{2006}): \emph{\bibinfo{title}{Proof Transformation by
  {CERES}}}.
\newblock In \bibinfo{editor}{Jonathan~M. \surnamestart Borwein\surnameend} \&
  \bibinfo{editor}{William~M. \surnamestart Farmer\surnameend}, editors: {\sl
  \bibinfo{booktitle}{5th International Conference on Mathematical Knowledge
  Management, {MKM}}}, {\sl \bibinfo{series}{Lecture Notes in Computer
  Science}} \bibinfo{volume}{4108}, \bibinfo{publisher}{Springer}, pp.
  \bibinfo{pages}{82--93}, \doi{10.1007/11812289_8}.

\bibitemdeclare{article}{Baaz2008CERES}
\bibitem{Baaz2008CERES}
\bibinfo{author}{Matthias \surnamestart Baaz\surnameend},
  \bibinfo{author}{Stefan \surnamestart Hetzl\surnameend},
  \bibinfo{author}{Alexander \surnamestart Leitsch\surnameend},
  \bibinfo{author}{Clemens \surnamestart Richter\surnameend} \&
  \bibinfo{author}{Hendrik \surnamestart Spohr\surnameend}
  (\bibinfo{year}{2008}): \emph{\bibinfo{title}{{CERES:} An analysis of
  F{\"{u}}rstenberg's proof of the infinity of primes}}.
\newblock {\sl \bibinfo{journal}{Theoretical Computer Science}}
  \bibinfo{volume}{403}(\bibinfo{number}{2-3}), pp. \bibinfo{pages}{160--175},
  \doi{10.1016/j.tcs.2008.02.043}.

\bibitemdeclare{article}{Baaz2012complexity}
\bibitem{Baaz2012complexity}
\bibinfo{author}{Matthias \surnamestart Baaz\surnameend},
  \bibinfo{author}{Stefan \surnamestart Hetzl\surnameend} \&
  \bibinfo{author}{Daniel \surnamestart Weller\surnameend}
  (\bibinfo{year}{2012}): \emph{\bibinfo{title}{On the complexity of proof
  deskolemization}}.
\newblock {\sl \bibinfo{journal}{Journal of Symbolic Logic}}
  \bibinfo{volume}{77}(\bibinfo{number}{2}), pp. \bibinfo{pages}{669--686},
  \doi{10.2178/jsl/1333566645}.

\bibitemdeclare{article}{Baaz2000Cut}
\bibitem{Baaz2000Cut}
\bibinfo{author}{Matthias \surnamestart Baaz\surnameend} \&
  \bibinfo{author}{Alexander \surnamestart Leitsch\surnameend}
  (\bibinfo{year}{2000}): \emph{\bibinfo{title}{Cut-elimination and
  Redundancy-elimination by Resolution}}.
\newblock {\sl \bibinfo{journal}{Journal of Symbolic Computation}}
  \bibinfo{volume}{29}(\bibinfo{number}{2}), pp. \bibinfo{pages}{149--177},
  \doi{10.1006/jsco.1999.0359}.

\bibitemdeclare{article}{Barbanera1996Symmetric}
\bibitem{Barbanera1996Symmetric}
\bibinfo{author}{Franco \surnamestart Barbanera\surnameend} \&
  \bibinfo{author}{Stefano \surnamestart Berardi\surnameend}
  (\bibinfo{year}{1996}): \emph{\bibinfo{title}{A Symmetric Lambda Calculus for
  Classical Program Extraction}}.
\newblock {\sl \bibinfo{journal}{Information and Computation}}
  \bibinfo{volume}{125}(\bibinfo{number}{2}), pp. \bibinfo{pages}{103--117},
  \doi{10.1006/inco.1996.0025}.

\bibitemdeclare{incollection}{Buss1995Herbrands}
\bibitem{Buss1995Herbrands}
\bibinfo{author}{Samuel~R. \surnamestart Buss\surnameend}
  (\bibinfo{year}{1995}): \emph{\bibinfo{title}{{O}n {H}erbrand's {T}heorem}}.
\newblock In: {\sl \bibinfo{booktitle}{Logic and Computational Complexity}},
  {\sl \bibinfo{series}{Lecture Notes in Computer Science}}
  \bibinfo{volume}{960}, \bibinfo{publisher}{Springer}, pp.
  \bibinfo{pages}{195--209}, \doi{10.1007/3-540-60178-3_85}.

\bibitemdeclare{techreport}{Cerna2018System}
\bibitem{Cerna2018System}
\bibinfo{author}{David \surnamestart Cerna\surnameend} \&
  \bibinfo{author}{Anela \surnamestart Lolic\surnameend}
  (\bibinfo{year}{2018}): \emph{\bibinfo{title}{System Description: {GAPT} for
  schematic proofs}}.
\newblock \bibinfo{type}{RISC Report Series}.
\newblock
  \urlprefix\url{http://www.risc.jku.at/publications/download/risc_5591/schematicGapt.pdf}.

\bibitemdeclare{article}{Eberhard2015Inductive}
\bibitem{Eberhard2015Inductive}
\bibinfo{author}{Sebastian \surnamestart Eberhard\surnameend} \&
  \bibinfo{author}{Stefan \surnamestart Hetzl\surnameend}
  (\bibinfo{year}{2015}): \emph{\bibinfo{title}{Inductive theorem proving based
  on tree grammars}}.
\newblock {\sl \bibinfo{journal}{Annals of Pure and Applied Logic}}
  \bibinfo{volume}{166}(\bibinfo{number}{6}), pp. \bibinfo{pages}{665--700},
  \doi{10.1016/j.apal.2015.01.002}.

\bibitemdeclare{article}{Ebner2018generation}
\bibitem{Ebner2018generation}
\bibinfo{author}{Gabriel \surnamestart Ebner\surnameend},
  \bibinfo{author}{Stefan \surnamestart Hetzl\surnameend},
  \bibinfo{author}{Alexander \surnamestart Leitsch\surnameend},
  \bibinfo{author}{Giselle \surnamestart Reis\surnameend} \&
  \bibinfo{author}{Daniel \surnamestart Weller\surnameend}
  (\bibinfo{year}{2018}): \emph{\bibinfo{title}{On the generation of quantified
  lemmas}}.
\newblock {\sl \bibinfo{journal}{Journal of Automated Reasoning}}, pp.
  \bibinfo{pages}{1--32}, \doi{10.1007/s10817-018-9462-8}.

\bibitemdeclare{misc}{GAPT2018User}
\bibitem{GAPT2018User}
\bibinfo{author}{Gabriel \surnamestart Ebner\surnameend},
  \bibinfo{author}{Stefan \surnamestart Hetzl\surnameend},
  \bibinfo{author}{Bernhard \surnamestart Mallinger\surnameend},
  \bibinfo{author}{Giselle \surnamestart Reis\surnameend},
  \bibinfo{author}{Martin \surnamestart Riener\surnameend},
  \bibinfo{author}{Marielle~Louise \surnamestart Rietdijk\surnameend},
  \bibinfo{author}{Matthias \surnamestart Schlaipfer\surnameend},
  \bibinfo{author}{Christoph \surnamestart Sp\"ork\surnameend},
  \bibinfo{author}{Janos \surnamestart Tapolczai\surnameend},
  \bibinfo{author}{Jannik \surnamestart Vierling\surnameend},
  \bibinfo{author}{Daniel \surnamestart Weller\surnameend},
  \bibinfo{author}{Simon \surnamestart Wolfsteiner\surnameend} \&
  \bibinfo{author}{Sebastian \surnamestart Zivota\surnameend}
  (\bibinfo{year}{2018}): \emph{\bibinfo{title}{{GAPT} user manual, version
  2.10}}.
\newblock \urlprefix\url{https://logic.at/gapt/downloads/gapt-user-manual.pdf}.

\bibitemdeclare{inproceedings}{Ebner2016System}
\bibitem{Ebner2016System}
\bibinfo{author}{Gabriel \surnamestart Ebner\surnameend},
  \bibinfo{author}{Stefan \surnamestart Hetzl\surnameend},
  \bibinfo{author}{Giselle \surnamestart Reis\surnameend},
  \bibinfo{author}{Martin \surnamestart Riener\surnameend},
  \bibinfo{author}{Simon \surnamestart Wolfsteiner\surnameend} \&
  \bibinfo{author}{Sebastian \surnamestart Zivota\surnameend}
  (\bibinfo{year}{2016}): \emph{\bibinfo{title}{System Description: {GAPT}
  2.0}}.
\newblock In \bibinfo{editor}{Nicola \surnamestart Olivetti\surnameend} \&
  \bibinfo{editor}{Ashish \surnamestart Tiwari\surnameend}, editors: {\sl
  \bibinfo{booktitle}{International Joint Conference on Automated Reasoning
  ({IJCAR})}}, {\sl \bibinfo{series}{Lecture Notes in Computer Science}}
  \bibinfo{volume}{9706}, \bibinfo{publisher}{Springer}, pp.
  \bibinfo{pages}{293--301}, \doi{10.1007/978-3-319-40229-1_20}.

\bibitemdeclare{article}{Furstenberg1955infinitude}
\bibitem{Furstenberg1955infinitude}
\bibinfo{author}{Harry \surnamestart Furstenberg\surnameend}
  (\bibinfo{year}{1955}): \emph{\bibinfo{title}{On the infinitude of primes}}.
\newblock {\sl \bibinfo{journal}{The American Mathematical Monthly}}
  \bibinfo{volume}{62}(\bibinfo{number}{5}), p. \bibinfo{pages}{353},
  \doi{10.2307/2307043}.

\bibitemdeclare{book}{Furstenberg1981Recurrence}
\bibitem{Furstenberg1981Recurrence}
\bibinfo{author}{Harry \surnamestart Furstenberg\surnameend}
  (\bibinfo{year}{1981}): \emph{\bibinfo{title}{Recurrence in ergodic theory
  and combinatorial number theory}}.
\newblock \bibinfo{publisher}{Princeton University Press},
  \doi{10.1515/9781400855162}.

\bibitemdeclare{article}{Gentzen1935Untersuchungen}
\bibitem{Gentzen1935Untersuchungen}
\bibinfo{author}{Gerhard \surnamestart Gentzen\surnameend}
  (\bibinfo{year}{1935}): \emph{\bibinfo{title}{Untersuchungen {\"{u}}ber das
  logische Schlie{\ss}en I}}.
\newblock {\sl \bibinfo{journal}{Mathematische Zeitschrift}}
  \bibinfo{volume}{39}(\bibinfo{number}{1}), pp. \bibinfo{pages}{176--210},
  \doi{10.1007/BF01201353}.

\bibitemdeclare{article}{Gentzen1936Widerspruchsfreiheit}
\bibitem{Gentzen1936Widerspruchsfreiheit}
\bibinfo{author}{Gerhard \surnamestart Gentzen\surnameend}
  (\bibinfo{year}{1936}): \emph{\bibinfo{title}{Die Widerspruchsfreiheit der
  reinen Zahlentheorie}}.
\newblock {\sl \bibinfo{journal}{Mathematische Annalen}} \bibinfo{volume}{112},
  pp. \bibinfo{pages}{493--565}, \doi{10.1007/BF01565428}.

\bibitemdeclare{article}{Gerhardy2005Extracting}
\bibitem{Gerhardy2005Extracting}
\bibinfo{author}{Philipp \surnamestart Gerhardy\surnameend} \&
  \bibinfo{author}{Ulrich \surnamestart Kohlenbach\surnameend}
  (\bibinfo{year}{2005}): \emph{\bibinfo{title}{Extracting Herbrand
  disjunctions by functional interpretation}}.
\newblock {\sl \bibinfo{journal}{Archive for Mathematical Logic}}
  \bibinfo{volume}{44}(\bibinfo{number}{5}), pp. \bibinfo{pages}{633--644},
  \doi{10.1007/s00153-005-0275-1}.

\bibitemdeclare{book}{Girard1987Proof}
\bibitem{Girard1987Proof}
\bibinfo{author}{Jean-Yves \surnamestart Girard\surnameend}
  (\bibinfo{year}{1987}): \emph{\bibinfo{title}{Proof theory and logical
  complexity}}.
\newblock \bibinfo{volume}{Vol. 1}, \bibinfo{publisher}{Bibliopolis}.

\bibitemdeclare{book}{Harper2016Practical}
\bibitem{Harper2016Practical}
\bibinfo{author}{Robert \surnamestart Harper\surnameend}
  (\bibinfo{year}{2016}): \emph{\bibinfo{title}{Practical foundations for
  programming languages}}.
\newblock \bibinfo{publisher}{Cambridge University Press},
  \doi{10.1017/CBO9781316576892}.

\bibitemdeclare{phdthesis}{Herbrand1930Recherches}
\bibitem{Herbrand1930Recherches}
\bibinfo{author}{Jacques \surnamestart Herbrand\surnameend}
  (\bibinfo{year}{1930}): \emph{\bibinfo{title}{Recherches sur la th{\'e}orie
  de la d{\'e}monstration}}.
\newblock Ph.D. thesis, \bibinfo{school}{Universit\'{e} de Paris}.

\bibitemdeclare{inproceedings}{Hetzl2012Applying}
\bibitem{Hetzl2012Applying}
\bibinfo{author}{Stefan \surnamestart Hetzl\surnameend} (\bibinfo{year}{2012}):
  \emph{\bibinfo{title}{Applying Tree Languages in Proof Theory}}.
\newblock In \bibinfo{editor}{Adrian-Horia \surnamestart Dediu\surnameend} \&
  \bibinfo{editor}{Carlos \surnamestart Mart{\'i}n-Vide\surnameend}, editors:
  {\sl \bibinfo{booktitle}{Language and Automata Theory and Applications}},
  {\sl \bibinfo{series}{Lecture Notes in Computer Science}}
  \bibinfo{volume}{7183}, \bibinfo{publisher}{Springer}, pp.
  \bibinfo{pages}{301--312}, \doi{10.1007/978-3-642-28332-1_26}.

\bibitemdeclare{inproceedings}{Hetzl2014Introducing}
\bibitem{Hetzl2014Introducing}
\bibinfo{author}{Stefan \surnamestart Hetzl\surnameend},
  \bibinfo{author}{Alexander \surnamestart Leitsch\surnameend},
  \bibinfo{author}{Giselle \surnamestart Reis\surnameend},
  \bibinfo{author}{Janos \surnamestart Tapolczai\surnameend} \&
  \bibinfo{author}{Daniel \surnamestart Weller\surnameend}
  (\bibinfo{year}{2014}): \emph{\bibinfo{title}{Introducing Quantified Cuts in
  Logic with Equality}}.
\newblock In \bibinfo{editor}{St{\'{e}}phane \surnamestart Demri\surnameend},
  \bibinfo{editor}{Deepak \surnamestart Kapur\surnameend} \&
  \bibinfo{editor}{Christoph \surnamestart Weidenbach\surnameend}, editors:
  {\sl \bibinfo{booktitle}{7\textsuperscript{th} International Joint Conference
  on Automated Reasoning, {IJCAR}}}, {\sl \bibinfo{series}{Lecture Notes in
  Computer Science}} \bibinfo{volume}{8562}, \bibinfo{publisher}{Springer}, pp.
  \bibinfo{pages}{240--254}, \doi{10.1007/978-3-319-08587-6_17}.

\bibitemdeclare{article}{Hetzl2014Algorithmic}
\bibitem{Hetzl2014Algorithmic}
\bibinfo{author}{Stefan \surnamestart Hetzl\surnameend},
  \bibinfo{author}{Alexander \surnamestart Leitsch\surnameend},
  \bibinfo{author}{Giselle \surnamestart Reis\surnameend} \&
  \bibinfo{author}{Daniel \surnamestart Weller\surnameend}
  (\bibinfo{year}{2014}): \emph{\bibinfo{title}{Algorithmic introduction of
  quantified cuts}}.
\newblock {\sl \bibinfo{journal}{Theoretical Computer Science}}
  \bibinfo{volume}{549}, pp. \bibinfo{pages}{1--16},
  \doi{10.1016/j.tcs.2014.05.018}.

\bibitemdeclare{article}{Hetzl2011CERES}
\bibitem{Hetzl2011CERES}
\bibinfo{author}{Stefan \surnamestart Hetzl\surnameend},
  \bibinfo{author}{Alexander \surnamestart Leitsch\surnameend} \&
  \bibinfo{author}{Daniel \surnamestart Weller\surnameend}
  (\bibinfo{year}{2011}): \emph{\bibinfo{title}{{CERES} in higher-order
  logic}}.
\newblock {\sl \bibinfo{journal}{Annals of Pure and Applied Logic}}
  \bibinfo{volume}{162}(\bibinfo{number}{12}), pp. \bibinfo{pages}{1001--1034},
  \doi{10.1016/j.apal.2011.06.005}.

\bibitemdeclare{inproceedings}{Hetzl2012Towards}
\bibitem{Hetzl2012Towards}
\bibinfo{author}{Stefan \surnamestart Hetzl\surnameend},
  \bibinfo{author}{Alexander \surnamestart Leitsch\surnameend} \&
  \bibinfo{author}{Daniel \surnamestart Weller\surnameend}
  (\bibinfo{year}{2012}): \emph{\bibinfo{title}{{T}owards {A}lgorithmic
  {C}ut-{I}ntroduction}}.
\newblock In: {\sl \bibinfo{booktitle}{Logic for Programming, Artificial
  Intelligence and Reasoning (LPAR-18)}}, {\sl \bibinfo{series}{Lecture Notes
  in Computer Science}} \bibinfo{volume}{7180}, \bibinfo{publisher}{Springer},
  pp. \bibinfo{pages}{228--242}, \doi{10.1007/978-3-642-28717-6_19}.

\bibitemdeclare{article}{Hetzl2013Expansion}
\bibitem{Hetzl2013Expansion}
\bibinfo{author}{Stefan \surnamestart Hetzl\surnameend} \&
  \bibinfo{author}{Daniel \surnamestart Weller\surnameend}
  (\bibinfo{year}{2013}): \emph{\bibinfo{title}{Expansion Trees with Cut}}.
\newblock {\sl \bibinfo{journal}{CoRR}} \bibinfo{volume}{abs/1308.0428}.
\newblock \urlprefix\url{https://arxiv.org/abs/1308.0428}.

\bibitemdeclare{book}{Hilbert1939Grundlagen}
\bibitem{Hilbert1939Grundlagen}
\bibinfo{author}{David \surnamestart Hilbert\surnameend} \&
  \bibinfo{author}{Paul \surnamestart Bernays\surnameend}
  (\bibinfo{year}{1939}): \emph{\bibinfo{title}{Grundlagen der Mathematik II}}.
\newblock \bibinfo{publisher}{Springer}.

\bibitemdeclare{article}{Leitsch2018Extraction}
\bibitem{Leitsch2018Extraction}
\bibinfo{author}{Alexander \surnamestart Leitsch\surnameend} \&
  \bibinfo{author}{Anela \surnamestart Lolic\surnameend}
  (\bibinfo{year}{2018}): \emph{\bibinfo{title}{Extraction of Expansion
  Trees}}.
\newblock {\sl \bibinfo{journal}{Journal of Automated Reasoning}}, pp.
  \bibinfo{pages}{1--38}, \doi{10.1007/s10817-018-9453-9}.

\bibitemdeclare{article}{Luckhardt1989Herbrand}
\bibitem{Luckhardt1989Herbrand}
\bibinfo{author}{Horst \surnamestart Luckhardt\surnameend}
  (\bibinfo{year}{1989}): \emph{\bibinfo{title}{{Herbrand-Analysen zweier
  Beweise des Satzes von Roth: Polynomiale Anzahlschranken}}}.
\newblock {\sl \bibinfo{journal}{Journal of Symbolic Logic}}
  \bibinfo{volume}{54}(\bibinfo{number}{1}), pp. \bibinfo{pages}{234--263},
  \doi{10.2307/2275028}.

\bibitemdeclare{article}{Miller1987Compact}
\bibitem{Miller1987Compact}
\bibinfo{author}{Dale~A. \surnamestart Miller\surnameend}
  (\bibinfo{year}{1987}): \emph{\bibinfo{title}{A compact representation of
  proofs}}.
\newblock {\sl \bibinfo{journal}{Studia Logica}}
  \bibinfo{volume}{46}(\bibinfo{number}{4}), pp. \bibinfo{pages}{347--370},
  \doi{10.1007/BF00370646}.

\bibitemdeclare{inproceedings}{Parigot1992Lambda}
\bibitem{Parigot1992Lambda}
\bibinfo{author}{Michel \surnamestart Parigot\surnameend}
  (\bibinfo{year}{1992}): \emph{\bibinfo{title}{$\lambda\mu$-Calculus: An
  Algorithmic Interpretation of Classical Natural Deduction}}.
\newblock In \bibinfo{editor}{Andrei \surnamestart Voronkov\surnameend},
  editor: {\sl \bibinfo{booktitle}{Logic for Programming, Artificial
  Intelligence and Reasoning (LPAR)}}, {\sl \bibinfo{series}{Lecture Notes in
  Computer Science}} \bibinfo{volume}{624}, \bibinfo{publisher}{Springer}, pp.
  \bibinfo{pages}{190--201}, \doi{10.1007/BFb0013061}.

\bibitemdeclare{inproceedings}{Pfenning1995Structural}
\bibitem{Pfenning1995Structural}
\bibinfo{author}{Frank \surnamestart Pfenning\surnameend}
  (\bibinfo{year}{1995}): \emph{\bibinfo{title}{Structural Cut Elimination}}.
\newblock In: {\sl \bibinfo{booktitle}{Logic in Computer Science}},
  \bibinfo{publisher}{{IEEE} Computer Society}, pp. \bibinfo{pages}{156--166},
  \doi{10.1109/LICS.1995.523253}.

\bibitemdeclare{book}{Takeuti1987Proof}
\bibitem{Takeuti1987Proof}
\bibinfo{author}{Gaisi \surnamestart Takeuti\surnameend}
  (\bibinfo{year}{1987}): \emph{\bibinfo{title}{Proof theory}},
  \bibinfo{edition}{second} edition.
\newblock {\sl \bibinfo{series}{Studies in Logic and the Foundations of
  Mathematics}}~\bibinfo{volume}{81}, \bibinfo{publisher}{North-Holland
  Publishing Co.}, \bibinfo{address}{Amsterdam}.

\bibitemdeclare{article}{Urban2001Strong}
\bibitem{Urban2001Strong}
\bibinfo{author}{Christian \surnamestart Urban\surnameend} \&
  \bibinfo{author}{Gavin~M. \surnamestart Bierman\surnameend}
  (\bibinfo{year}{2001}): \emph{\bibinfo{title}{Strong Normalisation of
  Cut-Elimination in Classical Logic}}.
\newblock {\sl \bibinfo{journal}{Fundamenta informaticae}}
  \bibinfo{volume}{45}(\bibinfo{number}{1-2}), pp. \bibinfo{pages}{123--155},
  \doi{10.1007/3-540-48959-2_26}.

\end{thebibliography}
\end{document}